\documentclass[pra,twocolumn,showpacs,amsmath,amssymb,superscriptaddress
,longbibliography
]{revtex4-1}

\usepackage[pdftex]{graphicx}
\usepackage[usenames,dvipsnames]{color}
\usepackage{epstopdf}
\usepackage[hidelinks]{hyperref}
\usepackage{tabularx}

\usepackage{mathrsfs}

\usepackage{amssymb,amsmath,amsthm}

\newtheorem{theorem}{Theorem}
\theoremstyle{definition}
\newtheorem{definition}{Definition}

\theoremstyle{remark}

\newtheorem{corollary}{Corollary}[theorem]
\newtheorem{lemma}[theorem]{Lemma}

\newcommand{\p}[1]{\mathbf{p}}
\newcommand{\q}[1]{\mathbf{q}}

\newcommand{\PP}[1]{\mathrm{(L)}}
\newcommand{\AP}[1]{\mathrm{(R)}}

\newcommand{\w}[1]{\mathbf{ w}}

\newcommand{\z}[1]{\theta}
\newcommand{\mix}[1]{\Omega}

\newcommand{\aop}[1]{{a}}
\newcommand{\cop}[1]{{c}}

\newcommand{\sm}[1]{\mathrm{SM}}
\usepackage{slashed}

\def\basiseps{{\bar\varepsilon}}

\usepackage[T1]{fontenc}
\usepackage{listings}

\lstdefinelanguage{Julia}%
  {morekeywords={abstract,break,case,catch,const,continue,do,else,elseif,%
      end,export,false,for,function,immutable,import,importall,if,in,%
      macro,module,otherwise,quote,return,switch,true,try,type,typealias,%
      using,while},%
   sensitive=true,%
   alsoother={$},%
   morecomment=[l]\#,%
   morecomment=[n]{\#=}{=\#},%
   morestring=[s]{"}{"},%
   morestring=[m]{'}{'},%
}[keywords,comments,strings]%

\graphicspath{{FIGs/}}

\begin{document}

\title{Bundled matrix product states represent low-energy excitations faithfully}
\author{Thomas E.~Baker}
\email[Please direct correspondence to: ]{bakerte@uvic.ca}
\affiliation{Department of Physics \& Astronomy, University of Victoria, Victoria, British Columbia V8P 5C2, Canada}
\affiliation{Department of Chemistry, University of Victoria, Victoria, British Columbia V8P 5C2, Canada}
\affiliation{Centre for Advanced Materials \& Related Technologies, University of Victoria, Victoria, British Columbia V8P 5C2, Canada}

\author{Negar Seif}
\affiliation{Department of Physics \& Astronomy, University of Victoria, Victoria, British Columbia V8P 5C2, Canada}

\date{\today}

\begin{abstract}

We consider a set of density matrices. All of which are written in the same orbital basis, but the orbital basis size is less than the total Hilbert space size. We ask how each density matrix is related to each of the others by establishing a norm between density matrices based on the truncation error in a partial trace for a small set of orbitals. We find that states with large energy differences must have large differences in their density matrices. Small energy differences are divided into two groups, one where two density matrices have small differences and another where they are very different, as is the case of symmetry. We extend these ideas to a bundle of matrix product states and show that bond dimension of the wavefunction ansatz for two states with large energy differences are larger. Meanwhile, low energy differences can have nearly the same bond dimensions for similar states.

\end{abstract}

\maketitle

\section{Introduction}

Density matrices represent one of the core objects in quantum mechanics. They store a wealth of information about the system and can be useful for solving problems. It is well established that the trace of the density matrix multiplied onto any operator gives the expectation value of the operator for a given state that the density matrix represents \cite{townsend2000modern,reif2009fundamentals}.

When diagonalized with an eigenvalue decomposition, density matrices are decomposed into a diagonal matrix that contains the orbital occupations of the natural orbitals. The natural orbitals themselves are the eigenvectors of the density matrix. It was originally pointed out by L\"owdin \cite{lowdin1955quantum,lowdin1955quantumb,lowdin1956natural} that natural orbitals were a rapidly converging basis set ({\it i.e.} the lowest eigenvalue converges faster than other choices of a basis with increasing numbers of natural orbitals). Because the solution of natural orbitals requires a ground state wavefunction, it is often computationally costly to obtain them before a computation. So, another basis set is often used.

However, describing the density matrix with a number of states equal to the total Hilbert space size is computationally costly. Reducing the number of degrees of freedom while maintaining accuracy on the result is the main challenge of computational chemistry and solutions of quantum problems on the classical computer in general. This is the foundational idea behind renormalization.

What is considered less often is how natural orbitals for one density matrix describe well or do not describe at all another density matrix. Our goal in this paper is to determine how density matrices when summed over some of the basis functions in a given basis can describe a system. We consider here the idea of bundling together different density matrices. We further consider how accurate those density matrices are if a common basis is used to write both density matrices in full. For example, if a set of $m$ orbtials that have the highest occupation for one density matrix are used to express another density matrix, how accurate can the second density matrix be and what is the best way to characterize it?

The fundamental question that is being asked here is how best can one relate what we will call a bundled set of density matrices, defined as follows.

\begin{definition}[Bundle of density matrices]\label{bundledef}
A bundle of density matrices is defined as a set of density matrices that are all written in the same basis.
\end{definition}

This is not the same as an ensemble of states contained in the density matrix.

The fundamental quantity that we want to investigate is whether a notion of closeness ({\it i.e.}, a norm) can be defined for the independent density matrices. The result used will be to establish a relationship between the truncation error and the metric distance between two density matrices. This will also be related to the energy difference between two states. The argumentation applies to any local Hamiltonian, which is reasonable for physical systems. We then extend the outcomes of those answers to matrix product states (MPS) to understand how the bond dimension of a bundle of MPSs will behave. This will explain why the bond dimension of the bundled MPS was not explosively large when an algorithm was formulated to solve for excitations in a quantum model \cite{baker2024direct,nakatani2014linear,huang2018generalized}.

The analysis tools used to formulate the truncation of the bundled MPS apply in principle to bundles of any type of density matrices so long as the eigenstates satisfy the area law of entanglement.  However, we choose to focus on bundles of eigenstates because they apply most readily in entanglement renormalization algorithms. In the following, we use theorems only when they are most relevant to the main thesis statement of the paper. We use definitions throughout to clearly define core concepts and keywords. Corollaries are used to communicate small extensions of the core theorems. Lemmas are used when heavy reliance on results outside of the paper are required to prove the statement and also when those statements are required for proofs later on. 

\section{Background on density matrices}

The class of problems that we wish to solve are based on the definition of a Hamiltonian that is a self-adjoint operator \cite{von1955mathematical}, $H$, composed of complex coefficients ($H\in \mathbb{C}^{M\times M}$ for an $M$ sized Hilbert space). The eigensolutions, $\psi$, of this operator have the relationship $H\psi=E\psi$ for eigenvalues (energies) $E\in\mathbb{R}$ and eigenvectors (wavefunctions) $\psi\in\mathcal{L}^2$ (square integrable) and contain complex numbers \cite{debnath2005introduction}. The Hamiltonian can be represented by a number of site indices $i,j,k,\ell,\ldots$. For many-body Hamiltonians, there is a quartic term (4-indices required) that appears to account for the electron-electron interaction, although the results we derive here will apply for any interaction. We choose to start from many-body Hamiltonians since this will recover a wide class a non-relativistic phenomena that we are interested in.

There is no consideration for divergences ({\it i.e.}, points where the evaluation of a quantity is infinite) in any terms as this analysis is solely concerned with models implemented to a finite numerical precision. Thus, all singularities in any interactions are regularized by finite difference approximations.

\subsection{Density matrices}

A density matrix can have several connotations. We explicitly define several that are useful here. The type of density matrix that will be used here is the one-body reduced density matrix, although we describe the more general case in many places.

\begin{definition}[Density matrix]
For a given Hamiltonian $H$, the full density matrix of a system is defined by
\begin{equation}\label{densitymatrix}
\rho=\sum_k\eta_k|\psi_k\rangle\langle\psi_k|
\end{equation}
for the $k$th excitation of the system and some occupation used throughout as $0\leq\eta_k\leq1$. When projected onto a real space lattice (or other basis) by through a resolution of the identity, the density matrix then assumes the form,
\begin{align}\label{realrho}
\rho=&\sum_{ ijk\mathcal{A}\mathcal{B}}\eta_k|i\mathcal{A}\rangle\langle i\mathcal{A}|\psi_k\rangle\langle\psi_k|j\mathcal{B}\rangle\langle j\mathcal{B}|\\
=&\sum_{ij\mathcal{A}\mathcal{B}}\rho_{ij\mathcal{A}\mathcal{B}}|i\mathcal{A}\rangle\langle j\mathcal{B}|
\end{align}
after defining $\rho_{ij\mathcal{A}\mathcal{B}}=\sum_k\eta_k\langle i\mathcal{A}|\psi_k\rangle\langle\psi_k|j\mathcal{B}\rangle$. When $\mathcal{A}$ and $\mathcal{B}$ are a null set, when either variable contains no indices, the density matrix is represented as a one-body reduced matrix ({\it i.e.}, only requiring the indices $i$ and $j$),
\begin{equation}
\rho=\sum_{ij}\rho_{ij}|i\rangle\langle j|.
\end{equation}
Had more indices been kept, then a higher order density matrix would be represented ({\it i.e.}, $i,j,k,\ell$ for the two-body reduced density matrix). In second quantization, one can simply compute $\rho_{ij}=\langle c^\dagger_{i\sigma} c_{j\sigma}\rangle$ (see Apx.~\ref{densitymatrixequivalence}) in a fermion model with spin $\sigma$. The density matrix defined in Def.~\ref{densitymatrix} is generally representing a {\it mixed density matrix}. 
\end{definition}

There are a few types of density matrices that form special types of the above definition. 

\begin{definition}[Pure state density matrix]
In the special case where  $\rho^2=\rho$ or that $\mathrm{Tr}(\rho^2)=1$ the density matrix is called a pure density matrix. Pure states can be represented as $\rho=|\psi\rangle\langle\psi|$ for some state $\psi$.
\end{definition}

In quantum chemistry, there is a different normalization convention. The trace of the density matrix is not always one, but instead is the number of particles, $N_e$, with a particular spin. We find it occasionally useful to refer to an ensemble density matrix that we define as the following. 

\begin{definition}[Ensemble density matrix]
Any density matrix whose trace is not one ($\eta_k$ can assume any positive value).
\end{definition}

Note that when the normalization of the ensemble density matrix has a trace that is not 1 that one can have $\mathrm{Tr}(\rho)\neq1$ and $\rho^2=\rho$. For example, three electrons of the same spin in a pure state can be represented in an ensemble density matrix and have $\mathrm{Tr}(\rho)=3$.

In many contexts, the mixed density matrix and the ensemble density matrix are used synonymously. A density matrix in the mixed representation can be thought of as a linear superposition of other states, so the concepts are the same. 

We define explicitly here the ensemble density matrix not only to distinguish between normalization factors in quantum information and quantum chemistry, but we also take this opportunity to highlight that the bundled density matrix from Def.~\ref{bundledef} is not an ensemble density matrix. Specifically, the bundle of density matrices can contain ensemble density matrices, although the density matrices in the bundle can be of any type. 

\begin{definition}[Natural orbitals]\label{naturalorbitalsdef}
The diagonalization of a density matrix yields a set of eigenvectors known as the {\it natural orbitals}. For the one-body reduced density matrix, these orbitals represent the density matrix as
\begin{equation}\label{densmat_NO}
\rho=\sum_k\varepsilon_k|\Phi_k\rangle\langle\Phi_k|
\end{equation}
with an eigenvalue $\varepsilon_k$ sometimes called an {\it occupational weight}.
\end{definition}

At first glance, Eq.~\eqref{densitymatrix} and Eq.~\eqref{densmat_NO} appear to be identical. However, this is not the case. The point of the definition of the natural orbitals is that they are expressed in a single-particle basis (only one coordinate $\mathbf{r}$ as used more extensively in Apx.~\ref{densitymatrixequivalence}). Meanwhile, the eigenvectors used in Eq.~\eqref{densitymatrix} have one coordinate $\mathbf{r}\in\mathbb{R}^3$ for each electron in the system.  

We will note that certain renormalization schemes can generate a more efficient basis set than natural orbitals \cite{bakerPRB18}. However, reducing the problem down to few enough orbitals that a polynomial time solver could be used would imply that the determination of that transformation is not discoverable in polynomial time since the general problem is known to be hard \cite{schuch2009computational}. Thus, no universally efficient procedure should be expected.

\begin{definition}[Expectation values]\label{defexpvalues}
For a given Hamiltonian operator $H$, the expectation value $E$ (energy) is
\begin{equation}\label{trace_expect}
E=\mathrm{Tr}(\rho H)\equiv\sum_k\langle k|\rho H|k\rangle.
\end{equation}
Replacing $H$ by any other operator $W$ gives the expectation value $\langle W\rangle$. The index $k$ is taken over any set of orbitals that is orthogonal and complete.
\end{definition}

Throughout, we will only consider orthogonal basis states, which applies equally to $k$ above, and the final result requires a local operator.

Typically what is done in practical computations is to take a truncated trace from the natural orbitals with the highest to the lowest weight. The result answer converges very quickly which can be seen from the using these functions in practice \cite{davidson1972properties,rothenberg1966natural,helbig2010physical,lathiotakis2010size,bakerPRB18}.

For completeness, we define an excitation in the system using the above concepts.

\begin{definition}[Excited states]
Given a set of excitations spanning an interval of energy, the next excitation can be defined as follows. To find an excitation at eigenenergy $\bar E$, take the set of all density matrices composed of eigenstates with $E<\bar E$. Then find the states that are orthogonal to those states. The minimum energy will be the excitation up to degeneracy.
\end{definition}

Two different excitations will have a different density matrices.

\subsection{Truncated density matrices}

So far, we have discussed density matrices where the basis states used to describe the density matrix spans the entire Hilbert space. Let us now define a truncated density matrix where small occupations are set to zero.
\begin{definition}[Truncated density matrix]
Returning to Def.~\eqref{naturalorbitalsdef}, the definition can be modified to define a {\it truncated density matrix} if the sum over $k$ in Eq.~\eqref{densmat_NO} is restricted to a value $m$ less than the Hilbert space size, $M$.  We denote the truncated trace as
\begin{equation}
\mathrm{Tr}_m^{(\gamma)}(\rho^{(\alpha)})=\sum_{i=1}^m\langle\Phi_i^{(\gamma)}|\rho^{(\alpha)}|\Phi_i^{(\gamma)}\rangle
\end{equation}
for orbitals from a set $\gamma$ on the $\alpha$ excitation.
\end{definition}

We use the term `truncated trace' because `partial trace' is usually associated with tracing over lattice sites and producing the partial density matrix.

There is an immediate consequence that the density matrix is now truncated, leading naturally to the definition of the truncation error.

\begin{definition}[Truncation error]\label{truncationerror}
The density matrix may be truncated to dimension $m$, known as the {\it bond dimension}. The difference from the true value of the trace of $\rho$ is known as the {\it truncation error}, $\delta$, which provides an estimate of the precision of the resulting expectation values from Def.~\ref{defexpvalues}. The full definition is then
\begin{equation}
\mathrm{Tr}(\rho^{(\alpha)})-\mathrm{Tr}_m^{(\alpha)}(\rho^{(\alpha)})=\sum_{i=1}^M\rho_{ii}^{(\alpha)}-\sum_{i=1}^m\rho_{ii}^{(\alpha)}\equiv\delta^{(\alpha)}_{m;\gamma}
\end{equation}
where $\alpha$ denotes a state that was used to construct $\rho$ and $\gamma$ is the basis over which the truncated trace was evaluated.
\end{definition}

\section{Relationship between two density matrices}

The energy difference between excitations can be defined using the expectation values as the following.

\begin{definition}[Energy difference]\label{Ediff}

For a given Hamiltonian $H$, the energy difference between two states $\alpha$ and $\beta$ is
\begin{equation}\label{energydiff}
\Delta E_{\alpha\beta} = \mathrm{Tr}\left((\rho^{(\alpha)}-\rho^{(\beta)})H\right)
\end{equation}
where $\Delta E_{\alpha\beta}=E_\alpha-E_\beta$. 
\end{definition}

Two excitations $\alpha$ and $\beta$ each have natural orbitals. These do not need to be the same set, nor orthogonal to each other. However, we restrict ourselves to orthogonal sets of natural orbitals so that they are related by a unitary transformation in every case.

The trace of any density matrix is invariant to the basis over which the trace is performed,
\begin{equation}\label{cyclictrace}
\mathrm{Tr}(\rho)=\mathrm{Tr}(U\rho U^\dagger)
\end{equation}
this is often known as the cyclic property of the trace. However, a unitary cannot be assumed at the start, and we will see that the traces between the two density matrices must be satisfied to ensure this is true.

\begin{lemma}[Relation between natural orbitals of different excitations]\label{NOrelation}
Natural orbitals of two different density matrices (assumed to be written in the same orbital basis with different occupational weights) are related by a unitary transformation if the trace is the same.
\end{lemma}

\begin{proof}
Two density matrices for two states $\alpha$ and $\beta$ satisfy
\begin{equation}
\mathrm{Tr}(\rho^{(\alpha)})-\mathrm{Tr}(\rho^{(\beta)})=\Delta N_e
\end{equation}
and by the cyclic property of the trace from Eq.~\eqref{cyclictrace}, we can transform $\rho^{(\alpha)}\rightarrow U\rho^{(\alpha)}U^\dagger$ without changing the trace. If the number of particles between the two density matrices is the same, then this implies that the natural orbitals between the two states are related by a unitary matrix. For two states with different particle number, we would need to use an orthogonal transformation $U^\dagger\rightarrow O^{-1}$ but the change in magnitude of the density matrix by the orthogonal transformation should be considered that we always consider natural orbitals with a normalized amplitude. Thus, one can enforce a normalization to again show that the natural orbitals are related by a unitary transformation.

\end{proof}

So long as the natural orbital sets for $\alpha$ and $\beta$ are both complete in the space spanned by two density matrices (including unoccupied natural orbitals), then this proof holds. We restrict our consideration to density matrices with the same total trace for simplicity, but all results can be extended to the case where an orthogonal transformation is required.

\subsection{Local systems}

The notion of locality can also apply to the eigensolutions and general operators \cite{von1955mathematical,cervera2017IPAM}. In fact, is a central idea in quantum physics for all types of physically relevant ground-states.

\subsubsection{Local correlations}

\begin{definition}[Locality]\label{localcorrelation}
The area law is a statement of correlations for extremal eigenvalues of the full spectrum. The two conditions for correlation functions (operator $O$) depending on whether there is a gap in the eigenvalue spectrum (gapped) or not (gapless) \cite{bakerCJP21,*baker2019m}:
\begin{equation}\label{arealaw_corr}
\langle i|O|j\rangle\sim\begin{cases}
\exp(-|i-j|/\xi)&\quad\mathrm{[gapped]}\\
1/|i-j|^\gamma&\quad\mathrm{[gapless]}
\end{cases}
\end{equation}
for two arbitrary, real exponents $\xi$ and $\gamma$. This condition holds up to a sign of the function. It is without loss of generality to more than two sites that this same definition still holds in the following.
\end{definition}

For a proof of the correlation dichotomy, in some contexts known as Kohn's near-sightedness principle, we refer the interested reader to Refs.~\onlinecite{hastings2004locality,kohn1996density,prodan2005nearsightedness}.

\begin{lemma}\label{localNOs}
The natural orbitals are local, defined as a non-zero on a compact subset of $\mathbf{r}\in\mathbb{R}^3$. 
\end{lemma}
\begin{proof}
Since they are derived from $\langle\hat c^\dagger_{i\sigma}\hat c_{j\sigma}\rangle$ for each element of the density matrix. Upon diagonalization, a linear combination of these elements will be the result. Since a linear combination of local correlation functions are local themselves, the natural orbitals are also local.
\end{proof}

\subsubsection{Local Hamiltonians}

The specific property of the system that we want to study here is for local Hamiltonians ({\it i.e.}~those with finite extent). The basic assumption at a coarse level constrains the long-range behavior of the Hamiltonian is contained in the following definition.
\begin{definition}[Local Hamiltonian]\label{localH}
A local Hamiltonian satisfies the following two properties in the thermodynamic limit.
\begin{equation}
\underset{|{i\mathcal{A}}-{j\mathcal{B}}|\rightarrow\infty}\lim\langle i\mathcal{A}|H|j\mathcal{B}\rangle\rightarrow0
\end{equation}
where $i$ and the set of coordinates $\mathcal{A}$. 
\end{definition}

\begin{definition}[Ultra-local Hamiltonian]\label{ultralocalH}
If the interactions are local, then in one limit 
\begin{equation}
\underset{|{i\mathcal{A}}-{j\mathcal{B}}|\rightarrow0}\lim\langle i\mathcal{A}|H|j\mathcal{B}\rangle\rightarrow C_{{i\mathcal{A}}}
\end{equation}
where $C_{{i\mathcal{A}}}\in\mathbb{C}$ is finite and real. This definition holds whether the sites $i\mathcal{A}$ and $j\mathcal{B}$ represent single sites or clusters of sites, but it represents the ultra local limit where the Hamiltonian appears truly local.
\end{definition}

\subsection{Relationship between natural orbital states}

It is useful to explore the relationship between the natural orbitals of two different states and how the unitary that connects them can appear. There are two broad categories that the unitary can take and it is worth explicitly defining each.  After defining the two cases, we remark on some physical cases where these can be found.

\begin{theorem}[Low truncation error implies a unitary transformation over relevant states between natural orbitals of two excitations]\label{UisIdentity}
A unitary matrix $U$ relating the natural orbitals of two excitations is nearly the identity except for a sub-block over $m$ states if the two states $\alpha$ and $\beta$ both have small truncation error in the $\mathrm m$ most important orbitals to $\rho^{(\alpha)}$.
\end{theorem}

\begin{proof}
Denote the occupation values of the rotated density matrix $U \rho^{(\beta)}U^\dagger$ as ($\mathrm{Tr}^{(\gamma)}_m$ denotes a truncated trace in a basis $\gamma$ for $m$ of the most relevant orbitals)
\begin{equation}\label{unitaryform}
\mathrm{Tr}^{(\alpha)}_m(U\rho^{(\beta)} U^\dagger)=\sum_{k=1}^m\basiseps_{kk}^{(\beta)}
\end{equation}
and the unrotated value would be the same but without the bar applied. The bar denotes the truncated trace in the space of the most relevant orbitals for $\alpha$. In general, this does not have to be the same evaluation as over the most important $m$ orbitals for $\beta$.

The following statement is equivalent to having a small truncation error in the most relevant orbitals for a state $\alpha$:
\begin{equation}\label{bigcondition}
\mathrm{Tr}(\rho^{(\beta)})-\mathrm{Tr}^{(\alpha)}_m(\rho^{(\beta)})=\delta^{(\beta)}_{m;\alpha}\overset!\cong\delta^{(\alpha)}_{m;\alpha}\ll\mathrm{Tr}(\rho^{(\alpha)})
\end{equation}
where it is clear that the truncation error is small. The trace is taken over the largest $1<m\lll M$ orbitals relevant for the most important natural orbitals for $\alpha$, $\{\Phi^{(\alpha)}\}_m$.  If we have $\delta^{(\beta)}_{m;\alpha}\cong\delta^{(\alpha)}_{m;\alpha}$,then the most important states for $\alpha$ give a small truncation error for $\beta$.

In order for the condition in Eq.~\eqref{bigcondition} to hold, the following must also be true
\begin{equation}\label{rotatedoccupation}
\sum_{k=1}^m\varepsilon_k^{(\alpha)}\approx\sum_{k=1}^m\basiseps_{kk}^{(\beta)}=\sum_{k,k'=1}^m\sum_{\ell=1}^M\varepsilon_{\ell}^{(\beta)}u_{\ell k}u_{\ell k' }^*\delta_{kk'}
\end{equation}
following Lemma~\ref{NOrelation} for an element of the unitary $u_{k\ell}$. We use the index $k$ for the basis of natural orbitals for $\alpha$. The indices $\ell$ and $\ell'$ for the natural orbitals for $\beta$.

There is a relationship between the unrotated coefficients and the $\beta$ state and those for the truncation in the states relevant for $\alpha$ (indexed by $k$),
\begin{equation}\label{rotateddensmat}
\sum_{k,k'=1}^m\sum_{\ell=1}^M\varepsilon_{\ell}^{(\beta)}u_{\ell k}u_{\ell k' }^*\delta_{kk'}\leq\sum_{\ell=1}^m\varepsilon_{\ell}^{(\beta)}
\end{equation}
where the sum $\ell$ over the diagonal elements of $\rho^{(\beta)}$ is taken to be over the most relevant $m$ natural orbitals for the $\beta$ state. The equality is true when the truncation error is small for both $\alpha$ and $\beta$ and only relevant orbitals for $\alpha$ are allowed, then those same orbitals must be relevant for $\beta$. Thus, the terms in the identity matrix appear strongly diagonal except for an $m\times m$ block for states of low truncation error.
\end{proof}

States where the unitary is close to the identity in the irrelevant $(M-m)\times(M-m)$ block of $U$ will be called {\it similar excitations}.

\begin{definition}[Similar excitations]\label{similarexc}
Two excitations have a similar set of $m$ natural orbitals if their density matrices have the relationship that a unitary $U$ transforming one natural orbital set into another has approximately the decomposition $U=W\oplus P$ for a unitary $W$ of size $m\times m$ and an second unitary matrix $P$ of size $(M-m)\times(M-m)$.
\end{definition}

This definition is not meant to be exhaustive or tight for all possible similar states, but it is sufficient for the states here.

The opposite would be {\it dissimilar excitations}.

\begin{definition}[Dissimilar excitations]
Two excitations have a similar set of $m$ natural orbitals but are not approximately of the form $U=W\oplus P$ as defined from similar states.
\end{definition}

This type of state is what would be encountered in the case of a symmetry protected state. Alternatively, two very widely separated centers of a potential $v(x)$ in an eigenvalue problem would also be dissimilar.

\subsection{Truncation errors as a metric distance}

The set of truncated density matrices should also be discussed in the context of a normed vector space, but the metric must be defined appropriately. In essence, we ask: how closely related are two truncated density matrices?

\begin{definition}[Relative truncation]
The relative truncation error between two states $\alpha$ and $\beta$ will be defined as
\begin{equation}\label{truncatedmetric}
r^{(\gamma)}_m(\alpha,\beta)=|\delta^{(\alpha)}_{m;\gamma}-\delta^{(\beta)}_{m;\gamma}|
\end{equation}
where it is very important to note that both truncation errors were evaluated in the same $m$-sized basis generated from excitation $\gamma$ (either $\alpha$ or $\beta$ or some other state).
\end{definition}

In order for the truncated set of natural orbitals from one excitation to be mapped onto a vector space with respect to the natural orbitals of another excitation, there must be a definition of the metric distance. We can show that the relative truncation is a suitable metric.

\begin{theorem}[Relative truncation error as a practical metric distance of truncated density matrices]\label{truncatedmetricdistance}
The relative truncation error can be used as a practical metric distance between states in almost all cases of practical interest. 
\end{theorem}
\begin{proof}
The absolute value of the truncation error, Eq.~\eqref{truncatedmetric}, satisfies all necessary properties of the metric distance in almost every case of relevance. The metric $r^{(\alpha)}_m$ for some truncation error $\delta^{(\alpha)}$ satisfies
\begin{enumerate}
\item $r^{(\gamma)}_m(\alpha,\alpha)=0$
\item $r^{(\gamma)}_m(\alpha,\beta)\geq0$
\item $r^{(\gamma)}_m(\alpha,\beta)=r^{(\gamma)}_m(\beta,\alpha)$
\item $r^{(\gamma)}_m(\alpha,\beta)\leq r^{(\gamma)}_m(\alpha,\zeta)+r^{(\gamma)}_m(\zeta,\beta)$ [Triangle ineq.]
\end{enumerate}
which are the properties of a metric \cite{dummit2004abstract}. The triangle inequality is satisfied in the same way that the $\mathcal{L}^1$-norm is constructed normally.  Points 1 and 2 follow by definition. However, it is technically possible that $r^{(\gamma)}_m(\alpha,\beta)$ can be equal to zero as is especially evident when $m=M$ as the truncation errors are both zero. Yet, it is generally the case that Point 2 holds for general states even when evaluating the truncation error out to numerical precision for $1\ll m\lll M$. This is why we call it a practical metric instead of simply a metric. The symmetry condition (3) is satisfied if the orbitals used to evaluate the truncation error is the same for both states. The triangle inequality (4) can be readily verified.
\end{proof}

The main takeaway from this identification of a metric is that the amount to which the two states differ (with respect to a common set of states $\gamma$) from each other is communicated through the truncation error. When the description of a given state is accurate to $\delta^{(\alpha)}$, then the states that will be most efficiently bundled are those also with low truncation error $\delta^{(\beta)}$.

The practical effect of this is that the truncation error is not only meaning the amount of information lost in a truncation to $m$ orbitals, but the difference between truncation errors is also communicating the distance between the excitations, a highly remarkable feature especially in the context of entanglement renormalization which normally only assigns the truncation error as a means of uncertainty. The metric here implies it can also determine how well an excitation is described in a given basis.

\subsection{Energy differences as a metric distance}

With regards to the identification of the metric for the truncation error, it is well-known that the truncation error implies an uncertainty in the energy, so it is often used as an error measure for entanglement renormalization methods \cite{dmrjulia1}. Because of this relationship between the truncation error and uncertainty on expectation values, it is natural to ask if another quantity can also serve as a metric distance, a quantity that is more physical.

The most readily available quantity in many cases is the energy difference. We seek to establish if the differences in energies can be proven to be a metric as well.

\begin{theorem}[Energies of truncated density matrices are a practical metric distance]\label{metric_diffE}
Two density matrices (for states $\alpha$ and $\beta$) truncated to order $m$ in some basis have a metric distance given by the absolute value of the energy difference, $|\Delta E^{(m)}_{\alpha\beta;\gamma}|$ with $\Delta E^{(m)}_{\alpha\beta;\gamma}=E^{(m)}_{\alpha;\gamma}-E^{(m)}_{\beta;\gamma}$, for local Hamiltonians expressed in local basis sets if not both of $E_\alpha$ and $E_\beta$ are zero.
\end{theorem}

\begin{proof}
Consider the energy difference in a given basis written as
\begin{equation}
\Delta E_{\alpha\beta;\gamma}^{(m)}=\left.\mathrm{Tr}_m^{(\gamma)}\left((\rho^{(\alpha)}-\rho^{(\beta)})H\right)\right|_{\gamma=\alpha}
\end{equation}
where we select the $\gamma=\alpha$ basis here for clarity but can select any basis without loss of generality. An equivalent way to express this is to take the trace over all real-space positions but only expand the density matrices out to order $m$ in the most relevant orbitals for $\alpha$. The resulting energy difference in the most relevant natural orbitals for $\alpha$ is then
\begin{align}\label{deltaEinbasis}
\Delta E_{\alpha\beta;\alpha}^{(m)}=&\sum_{i\mathcal{A}}\langle i\mathcal{A}|\sum_{k=1}^m\left(\varepsilon^{(\alpha)}_k|\Phi^{(\alpha)}_k\rangle\langle\Phi^{(\alpha)}_k|\right.\\
&\left.-\sum_{\ell=1}^M\sum_{k'=1}^m\varepsilon^{(\beta)}_\ell u_{\ell k}u_{\ell k' }^*|\Phi^{(\alpha)}_k\rangle\langle\Phi^{(\alpha)}_{k'}|\right)H|i\mathcal{A}\rangle\nonumber
\end{align}
where $\rho^{(\beta)}=\sum_\ell\varepsilon_\ell^{(\beta)}|\Phi_\ell^{(\beta)}\rangle\langle\Phi_\ell^{(\beta)}|$ with $\ell$ indexing the complete space and with the unitary defined earlier in Eq.~\eqref{unitaryform}.

 Let the definition
\begin{equation}
\basiseps_{kk'}^{(\beta)}=\sum_{\ell=1}^M\varepsilon^{(\beta)}_\ell u_{\ell k}u_{\ell k' }^*
\end{equation}
hold in the following. A complete set of states can be inserted before $H$ to give
\begin{equation}
H|i\mathcal{A}\rangle\rightarrow\sum_{j\mathcal{B}}|j\mathcal{B}\rangle\langle j\mathcal{B}| H|i\mathcal{A}\rangle
\end{equation}
where the term $\langle i\mathcal{A}|H|j\mathcal{B}\rangle$ can be thought of as the Hamiltonian in real space. For local Hamiltonians, this term goes to a delta function for large differences between the sites $i,j,\mathcal{A},\mathcal{B}$ denoted as $h_{ij\mathcal{A}\mathcal{B}}$. 

In the local limit where $(h_{ij\mathcal{A}\mathcal{B}}\rightarrow C_{ {i\mathcal{A}}}\delta_{i\mathcal{A},j\mathcal{B}}$), Eq.~\eqref{deltaEinbasis} becomes
\begin{align}
\Delta E_{\alpha\beta}^{(m)}=&\sum_{{i\mathcal{A}}}\langle {i\mathcal{A}}|\sum_{k=1}^m\left(\varepsilon^{(\alpha)}_k|\Phi^{(\alpha)}_k\rangle\langle\Phi^{(\alpha)}_k|\right.\\
&\left.-\sum_{k'=1}^m\basiseps^{(\beta)}_{kk'}|\Phi^{(\alpha)}_k\rangle\langle\Phi^{(\alpha)}_{k'}|\right)|{i\mathcal{A}}\rangle C_{{i\mathcal{A}}}\nonumber
\end{align}

At this point, the density matrix for $\beta$ projected into the orbitals for $\alpha$ is not diagonal. If we impose that the orbital basis is local, as in Lemma~\ref{localNOs} for natural orbitals, then this implies that the as the difference between $k$ and $k'$ are large, then the overlap between functions must go to zero. Recall that we only consider orthogonal basis sets. This means that $\sum_\ell u_{\ell k}u_{\ell k'}^*=\delta_{kk'}$ in the ultra-local limit. Consequently, the expression reduces to 
\begin{align}\label{deltaEmetric}
\Delta E_{\alpha\beta}^{(m)}=&\sum_{{i\mathcal{A}}}\sum_{k=1}^mC_{{i\mathcal{A}}}\left(\varepsilon^{(\alpha)}_k-\basiseps^{(\beta)}_{k}\right)|\langle\Phi^{(\alpha)}_k|{i\mathcal{A}}\rangle|^2.
\end{align}
where $\basiseps_k^{(\beta)}$ is the diagonal coefficients of $\rho^{(\beta)}$ in the basis of the most relevant orbitals for $\alpha$. 

With an absolute value sign, $|\Delta E^{(m)}_{\alpha\beta;\gamma}|$, one can verify the axioms of a metric for $1\ll m\lll M$ are practically satisfied as presented in Thm.~\ref{truncatedmetricdistance}. 

There can be an ambiguity in the definition of the metric if the states are zero but this is not typical of eigenvalue problems relevant here. So, we exclude this case.
\end{proof}

Theorem~\ref{metric_diffE} made heavy use of the ultra-local limit where the Hamiltonian is effectively diagonal. This and several other features are discussed here as useful to derive the core statements, but understanding how things would change if the Hamiltonian is longer-range is worth performing. 

If the Hamiltonian contains longer range terms, then terms like $\langle i\mathcal{A}|H|j\mathcal{B}\rangle$ to some relative distance between $i\mathcal{A}$ and $j\mathcal{B}$ can be incorporated.

Many common basis functions decay exponentially. Notably, the Gaussian basis set decays exponentially away from the origin of the Gaussian. Thus, in the regime where matrix product states are best applied (local, gapped Hamitonians), the local approximation is a very good one for the excitations.

The ultra-local limit is not the most general form that the unitary can take here for arbitrary problems, but it is most applicable one for the tensor network case we find below where there is merely a permutation of the elements in the same basis. All of the excitations will be written in the same basis of entangled states and the density matrix eigenvalues of each is in an $m\times m$ basis ({\it i.e.}, effectively all of the occupational weights can be discovered in $g$ same-size matrices which all are diagonal and in the same basis). If one prefers, then $\{\Phi_k^{(\alpha)}\}$ and $\{\Phi_k^{(\beta)}\}$ belong to a common set $\gamma$ and an $m\times m$ tensor with occupational weights can be identified for both states in that same basis.

\subsection{Large energy differences}

We can now make a general statement about large energy differences.

\begin{corollary}[States with large energy differences]
Large energy differences imply larger differences between density matrices.
\end{corollary}
\begin{proof}
Note that
\begin{equation}\label{deltaEineq}
|\Delta E_{\alpha\beta;\gamma}^{(m)}|\leq\left|\mathrm{Tr}_m\Big(|\rho^{(\alpha)}-\rho^{(\beta)}|H\Big)\right|\leq\mathrm{Tr}\Big(|\rho^{(\alpha)}-\rho^{(\beta)}|H\Big)
\end{equation}
where $|\rho^{(\alpha)}-\rho^{(\beta)}|$ is the element-wise absolute value implemented in the following way
\begin{equation}
|\rho^{(\alpha)}-\rho^{(\beta)}|=\sum_{k,\ell=1}^m\left|\tilde \rho_{k\ell}^{(\alpha)}-\tilde\rho_{k\ell}^{(\beta)}\right||\Phi^{(\gamma)}_k\rangle\langle\Phi^{(\gamma)}_\ell|
\end{equation}
where tildes denote that the density matrices are rotated into the $\gamma$ basis. Effectively, one simply replaces round braces around the occupation values in Eq.~\eqref{deltaEmetric} by an absolute value. By comparison with Eq.~\eqref{deltaEmetric}, then Eq.~\eqref{deltaEineq} is satisfied. Further, recall that coefficients $C_{{i\mathcal{A}}}$ in Eq.~\eqref{deltaEmetric} set an effective energy scale. The maximum of which, $C_\mathrm{max.}$, can act as a normalization to give
\begin{equation}
\frac{|\Delta E_{\alpha\beta;\gamma}^{(m)}|}{C_\mathrm{max.}}\leq\frac{\left|\mathrm{Tr}_m(|\rho^{(\alpha)}-\rho^{(\beta)}|H)\right|}{C_\mathrm{max.}}\leq|\rho^{(\alpha)}-\rho^{(\beta)}|_F
\end{equation}
where$|\rho^{(\alpha)}-\rho^{(\beta)}|_F$ is the Froebenius norm and since coefficients of $H$ can be either positive or negative. This is true by inspection of Eq.~\eqref{deltaEmetric} and $|\rho^{(\alpha)}-\rho^{(\beta)}|_F=\sum_{k,\ell=1}^m|\rho_{k\ell}^{(\alpha)}-\rho_{k\ell}^{(\beta)}|$. The sum over $k$ and $\ell$ in the Froebenius norm can be truncated to $m$ or kept in the full basis.

Summarily, when $|\Delta E_{\alpha\beta;\gamma}^{(m)}|$ is large, then the density matrices must be different as conveyed by the Froebenius norm. So, a large energy difference implies a large difference between the states $\alpha$ and $\beta$.
\end{proof}

The result is valid in any basis if the operator is local. The use of a finite number of states $m$ for the orbital basis makes this result useful in a truncated space.

So, density matrices with large energy differences necessarily have large differences. The opposite is not so clearly defined. If the energies are low, they can either be a similar state, in which the most relevant basis for one excitation is relevant for another state. Alternatively, if the state is dissimilar, then the basis relevant for both excitations is very different, as discussed earlier.

There are several useful consequences of the previous derivation of the main result. We will discuss them here before moving towards the solution with partial density matrices.

\begin{corollary}[Transitive property for many states]\label{manystates}
The results of Thm.~\ref{metric_diffE} generalize to all nearby ($\Delta E_{\alpha\beta}\approx0$) states with low truncation error.
\end{corollary}

\begin{proof}[Proof]
Reconsider the form of Eq.~\eqref{deltaEmetric} but now for a chain of energy differences between several states of a single symmetry sector. The first energy difference satisfies Eq.~\eqref{deltaEmetric} as does the next set of two excitations. Therefore, as a transitive property, all wavefunctions with small energy differences share a high degree of overlap in their truncated density matrices.
\end{proof}

Taken together, the results of this section show that partial density matrices with low energy can be bundled efficiently ({\it i.e.}, a basis $m$ can give a small truncation error for both states) together. Each excitation added to the bundle comes at either no additional cost in terms of the number of natural orbitals required for a small truncation error if the states are similar. For small energy differences, there is a possibility that the small energy states cost no more than making the lowest energy state. When bundling two states in the bulk of the eigenvalue spectrum together, one must only pay the cost of one of the states, the second comes at a small cost if similar.

When adding a dissimilar state to the bundle, a sufficient number of natural orbitals must be added to the orbital basis to give a low truncation error for that dissimilar state. Once the new dissimilar state is added to the bundle, a new set of nearby similar states can be bundled at small cost.

\subsection{Application to quantum chemistry}

The primary focus of this paper is on locally entangled systems. This is not the general case in quantum chemistry where models have longer range interaction. In this case, the entanglement of the states is known to be much larger than for the local models as can be seen from direct computation. In particular, we note that the extension of models to two dimensions can cause an exponential increase in the entanglement with the width of the system \cite{schollwock}.

This would cast doubt on whether the ultra-local limit will apply in a quantum chemistry system. It turns out that by use of the singular value decomposition (SVD), the basis states of the model can be written such that the basis functions between two states are identical for two different systems. For example, writing two excitations and then cutting the system with the SVD at the same bond in each system will give a form like $\Psi=UDV^\dagger$ which has grouped basis functions to the left of the cut (contained in $U$) and the other basis functions for the right of the cut (contained in $V^\dagger$). Thus, the same basis functions can be assigned for both $U$ and $V^\dagger$ for both states. Between the two states, the basis functions may be different in practice, but we can always find a unitary transformation that makes the states match between the $U$ ($V^\dagger$) for one state and the other.

However, the two $D$ matrices are still diagonal. This is exactly the ultra-local limit as we applied it to the density matrix (but not the Hamiltonian from Def.~\ref{ultralocalH}) because the only difference is that the two states have different occupation numbers inside of $D$. Thus, only a permutation is allowed from the basis functions to rearrange the occupation numbers of the first state to the second. Thus, the ultra-local limit, in the construction of the SVD, is valid even between two different states. We merely motivated the limit by applying an understanding of the real-space orbitals previously.

Thus, quantum chemistry can fit into this hierarchy established here. The computational bottleneck appears in retaining enough states in the SVD in the quantum chemistry system. Because entanglement is larger, there must be more states retained, thus the method is less efficient.

\section{Density matrix renormalization}
Up to now, all considerations have been for the full wavefunction and full density matrix and $n$-body reduced density matrices. The statement for the full wavefunction is not completely useful for the solution in a tensor network decomposition where partial density matrices on either the left or right partition of a given system (partitioned in the sense of partitioning a graph) is the relevant quantity for the eventual solution. 

\begin{definition}[Partial density matrix]
A system partitioned into two sets can identify basis functions for each set. The basis functions may be used to project the full density matrix into a reduced site representation. The density matrix on one half of the system is equivalent to the full density matrix but traced out on the other half. No matter how the system is partitioned, the entanglement expressed by one partial density matrix must be equal to the complementary density matrix's entanglement.
\end{definition}

We do not rule out that a null set can be used for one of the two partitions, but this would imply that the partial density matrix is equivalent to the full density matrix. The sites do not need to be contiguous, but they are chosen to be such here.

\subsection{Matrix product states}

An entanglement renormalization algorithm explicitly calculates the components of a density matrix partitioned between two parts of a system with the use of a singular value decomposition (see Ref.~\onlinecite{bakerCJP21,*baker2019m,dmrjulia1} for an explicit derivation). The following definitions cover the matrix product state and make use of the above theorems.

\begin{definition}[Matrix product state]
A wavefunction with degrees of freedom $\sigma_i$ on each site $i$ is written as
\begin{equation}\label{wavefctdef}
|\psi\rangle=\sum_{\{\sigma_i\}}c_{\sigma_1\sigma_2\sigma_3\sigma_4\ldots}|\sigma_1\sigma_2\sigma_3\sigma_4\ldots\rangle
\end{equation}
for some probability amplitudes $c_{\sigma_1\sigma_2\sigma_3\sigma_4\ldots}\in\mathbb{C}$. By performing a series of reshapes and SVDs \cite{dmrjulia1}, Eq.~\eqref{wavefctdef} can be decomposed into a series of tensors \cite{bakerCJP21,*baker2019m}
\begin{equation}
|\psi\rangle=\sum_{\substack{\{\sigma_i\}, \{a_i\}}}
A^{\sigma_1}_{a_1}A^{\sigma_2}_{a_1a_2}\ldots D^{\sigma_i}_{a_{i-1}a_i}\ldots 
B^{\sigma_{N}}_{a_{N-1}}|\sigma_1\ldots\sigma_N\rangle
\end{equation}
for a number of sites, $N$. Raised and lowered indices are for ease of viewing as there is no notion of covariant or contravariant indices \cite{schutz2009first}. Raised and lowered indices are used here to signify the partitioning of the lattice via a reshaping operation \cite{bakerCJP21,*baker2019m}. The index introduced by the SVD ($a_i$) is known as a link index and its dimension is called a {\it bond dimension}. The orthogonality center $D$ contains the weight of the basis sets and can be gauged to any site or bond \cite{dmrjulia1}. All tensors to the left of the center of orthogonality ($A$) are {\it left-normalized}. Similarly, all tensors to the right of the orthogonality center ($B$) are said to be {\it right-normalized}. Contraction of any left- or right-normalized matrix with itself leaving only two indices corresponding to $a_i$ uncontracted yields an identity matrix.
\end{definition}

Eq.~\eqref{wavefctdef} can be represented with Penrose's graphical notation as in Fig.~\ref{MPSfig} \cite{penrose1971angular}. Vertical lines on the MPS correspond to the $\sigma$ indices and $a$ are the horizontal indices.

Because of how the density matrices are truncated in the MPS through the $D$ matrix retaining the largest values, we can remark that the basis states that are kept represent the most entangled basis functions between the two partitions. Thus, the use of natural orbitals in the previous derivations and theorems now becomes a set of the most entangled basis functions between the two partitions.

\subsection{Matrix product states represent the ground-state faithfully}

In a seminal work in Ref.~\onlinecite{verstraete2006matrix}, it was demonstrated that the MPS can represent the true ground state with a low error. The rigorous result is established with a bound on the Renyi entropy in relation to the truncation error of the MPS. The argument demonstrates that locally entangled models are efficiently described by the MPS ansatz and there is a generalization to the MPS in higher dimensions and connections to the area law of entanglement.

The techniques used in this paper veer closer to those used in quantum chemistry, making use of the MPS's relationship to the full density matrix. We recast the results of Ref.~\onlinecite{verstraete2006matrix} into the tools used here in case it is useful.

Consider a wavefunction of $N$ sites. Reshaping the $N$ degrees of freedom into two groups, a left group and a right group allows us to take a singular value decomposition of the form (we drop the dagger from $V$ when writing it in terms of tensor components for clarity)
\begin{equation}
\psi_{\sigma_1\ldots\sigma_N}=U^{\sigma_1\ldots\sigma_j}_{a_{j-1}}D_{a_{j-1}a_j}V_{a_j}^{\sigma_{j+1}\ldots\sigma_N}
\end{equation}
where $\sigma_i$ represents the degrees of freedom locally on each site and the wavefunction was decomposed on the $j$th bond. Raised and lowered indices do not mean anything and are only used for clarity.

The decomposition according to the SVD gives the elements necessary to construct the partial density matrix for the left
\begin{equation}
\rho_L=UD^2U^\dagger
\end{equation}
and right
\begin{equation}
\rho_R=VD^2V^\dagger
\end{equation}
of the system \cite{bakerCJP21,*baker2019m,dmrjulia1}. The occupations of the natural orbitals are contains in the $D^2$ matrices and are known to decay rapidly. Note that if we had the expectation value of the Hamiltonian, $H$,
\begin{equation}
E=\mathrm{Tr}(\rho H)=\sum_i\lambda_ih_{ii}
\end{equation}
where the basis of the natural orbitals was used to make the density matrix diagonal. If the orbital occupations of the density matrix are $\rho_{ii}=\lambda_i$ and ordered as $1\geq\lambda_1\geq\lambda_2\geq\lambda_3\geq\lambda_4\ldots\geq\lambda_N\geq0$, and decay rapidly, which is an assumption we use throughout.

Because the occupation values of the density matrix decay rapidly, the entanglement, $S=-\mathrm{Tr}(\rho\ln\rho)$ is well approximated. Thus, if the entanglement is low, then the truncation error is also low \cite{verstraete2006matrix}.

This argument we use here leaves out the extensions to the area law of entanglement, but we find these methods useful for the extension to the bundled MPS case.

\subsection{Bundled matrix product states}

The bundled MPS represents several density matrices considered together but independently. We note a key distinction that when writing a wavefunction for two excitations $\Psi=(\psi_1+\psi_2)/\sqrt2$ that the states individually would form a pure density matrix but that the combination is mixed. The density matrices for each state in the bundle would still be pure. Since the wavefunctions are considered independent from each other is the key difference between the bundled and ensemble density matrices.  The most common form will be a set of mixed density matrices (with or without normalization of the trace to a value of 1).

\begin{figure}[t]
\includegraphics[width=\columnwidth]{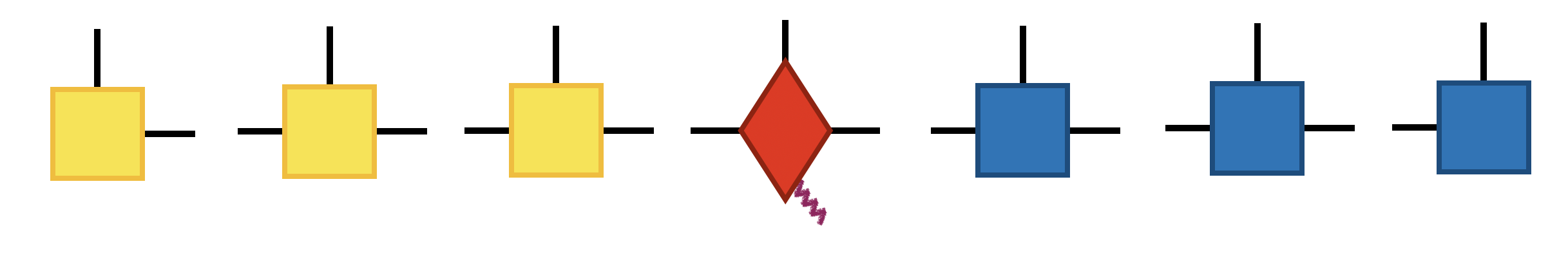}
\caption{The matrix product state using the graphical notation owed to Penrose \cite{bakerCJP21,*baker2019m}. A bundled MPS has the extra index on the orthogonality center while the regular MPS does not.\label{MPSfig}}
\end{figure}

\begin{definition}[Bundled matrix product state]
A ensemble of excitations can be added to the MPS with the additional of another index $\xi$ onto the orthogonality center.
\begin{equation}\label{ensembleMPSdef}
|\Psi\rangle=\sum_{\substack{\{\sigma_i\},\\ \{a_i\}, \xi}}
A^{\sigma_1}_{a_1}A^{\sigma_2}_{a_1a_2}\ldots D^{\sigma_i\xi}_{a_{i-1}a_i}\ldots
B^{\sigma_{N}}_{a_{N-1}}|\sigma_1\ldots\sigma_N\rangle|\xi\rangle
\end{equation}
where $\psi_\xi$ represents $g\in\mathbb{N}$ excitations indexed by $\xi$. This is also referred to as a bundled MPS.
\end{definition}

\subsection{Construction of the bundled matrix product state}

An explicit construction of the exact bundled MPS was demonstrated in Ref.~\onlinecite{baker2024direct} from a group of MPSs. Note that that construction in Ref.~\onlinecite{baker2024direct} over-determines the form of the left- and right-normalized matrices. The ability to compress the bundled MPS is the main point illuminated by the results in this paper ({\it i.e.}, many states are similar enough to reduce the size of the left- and right-normalized tensors).

An alternative derivation to Ref.~\onlinecite{baker2024direct} can be found from the traditional form of a set of excitations as an $d^{N}\times g$ matrix for $N$ sites, size of the physical index $d$, and number of excitations $g$, giving a complex coefficient, $c_{\sigma_1\sigma_2\ldots\sigma_{N}}^\xi$ where $\sigma_i$ is defined as in the main text to index the physical degrees of freedom on a site $i$ and $\xi$ indexes the excitations.

The task to separate the indices follows the same procedure as the derivation of the MPS from a single wavefunction \cite{bakerCJP21,*baker2019m}. The first step is to reshape the matrix to isolate $\sigma_1$ from the rest of the indices and perform a singular value decomposition.
\begin{equation}
|\Psi\rangle=\sum_{\substack{\sigma_1\ldots\sigma_{N} \\ a_1\xi}}A^{\sigma_1}_{a_1}c^{\xi\sigma_2\ldots\sigma_{N}}_{a_1}|\sigma_1\ldots\sigma_{N}\rangle|\xi\rangle
\end{equation}
Continuing with more SVDs and reshapes, we discover more terms in the bundled MPS until recovering Eq.~\eqref{ensembleMPSdef}. 

The importance of this exercise is to discover the size of the bond dimension at the ends of the bundled MPS.  Recall that for the single MPS that the maximum size of the bond dimension in the exact case (without truncation) scales as
\begin{equation}\label{maxm}
m_{\max}=\min\left(\prod_{x=1}^{i}d_x,\prod_{x=i+1}^{N_s}d_x\right)
\end{equation}
for bond $i$ and physical index of size $d_x$ on bond $x$. This is because the SVD is chosen with a convention that the size of the $D$ matrix is the minimum of the two dimensions of the input matrix \cite{dmrjulia1}. 

A similar expression can be derived for the bundled MPS. The only difference is that when gauged to a site $j$, the excitation index can be thought of as multiplying the bond dimension for the purposes of this size counting argument. Thus, redefining Eq.~\eqref{maxm} with $d_j\rightarrow g\cdot d_j$ will be sufficient.

The main point of this exercise is to show that when the orthogonality center is gauged to the first or last site in the system, one half of the system has the same size bond dimensions as the single MPS. The other half of the system has a larger bond dimension dependent on $g$.

This procedure is unfeasible beyond small systems since the exact wavefunctions can not be determined efficiently by exact diagonalization from the exponential size of the Hamiltonian operator with increasing number of sites \cite{bakerCJP21,*baker2019m}.

\begin{figure}[t]
\includegraphics[width=\columnwidth]{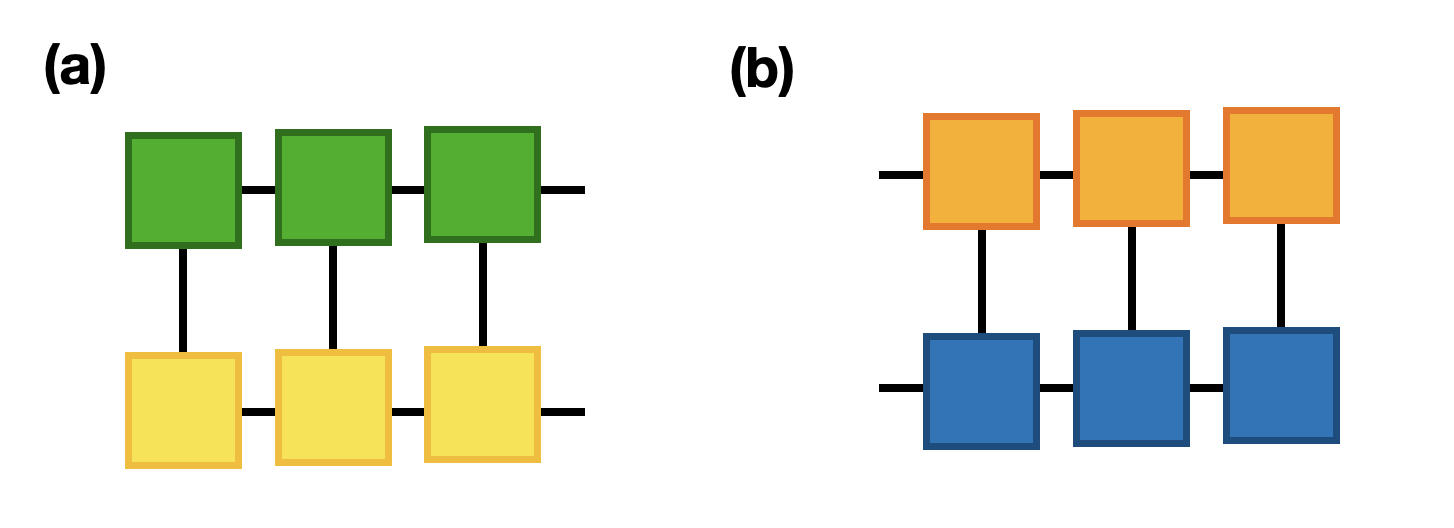}
\caption{The partial density matrices formed from the (a) left- and (b) right-normalized tensors of the MPS for a state $\psi$ (lower tensors) and an excitation $\phi$ (upper tensors). Connected lines signify a contraction over the relevant indices \cite{bakerCJP21,*baker2019m}. The $r$-body reduced density matrix is found by contracting the open link indices (horizontal) in the system and not contract $r$ of the physical indices (vertical). \label{partialdensmat}}
\end{figure}

\subsection{Overlap matrices of system partitions}\label{partialdensmats}

The overlap matrices representing the left and right partitions of the system can be constructed from the left- and right-normalized tensors of the MPS.

\begin{definition}[Overlap matrices of partial density matrices]\label{partialdens_def}
The wavefunction $\psi$ may be decomposed with an SVD as (with the matrix $D$ given on a bond, not a site)
\begin{equation}\label{wavefctSVD}
|\psi\rangle=U^{\sigma_1\sigma_2\ldots\sigma_i}_{a_{i-1}}D_{a_{i-1}a_i}V^{\sigma_{i+1}\ldots\sigma_N}_{a_i}|\sigma_1\ldots\sigma_N\rangle
\end{equation}
The contraction of the wavefunction onto another state $\phi$ is given by
\begin{equation}
\langle\phi|\psi\rangle=\sum_{a_{i-1},a_i,a'_{i-1},a'_i}\rho^{(L)}_{a_{i-1}a_{i-1}'}D_{a_{i-1}a_i}\tilde D_{a_{i-1}'a_i'}\rho^{(R)}_{a_ia_i'}
\end{equation}
and with
\begin{align}
\rho^{(L)}_{a_{i-1}a_{i-1}'}=\sum_{\{\sigma_i\}}\tilde U^{\sigma_1\sigma_2\ldots\sigma_i}_{a_{i-1}'}U^{\sigma_1\sigma_2\ldots\sigma_i}_{a_{i-1}}\\
\rho^{(R)}_{a_ia_i'}=\sum_{\{\sigma_i\}}V^{\sigma_{i+1}\ldots\sigma_N}_{a_i}\tilde V^{\sigma_{i+1}\ldots\sigma_N}_{a_i'}
\end{align}
with both of those terms represented in Fig.~\ref{partialdensmat}. The matrices $D$ can be truncated just as the density matrix was to $m\leq M$ just as in a truncated SVD. This will define a {\it truncated partial density matrix}.
\end{definition}

The following minor statement is presented as a lemma because it directly follows from the additional definition of the partial density matrix and Thm.~\ref{metric_diffE}.

\begin{lemma}[Commonality of singular vectors for different excitations]\label{SVDs_overlap}
A set of relevant singular vectors for a partial density matrix of an excitation can differ at most from the most relevant singular vectors of another excitation related to their energy difference similar to Thm.~\ref{metric_diffE}'s statement for natural orbitals.
\end{lemma}

\begin{proof}
The partial density matrices are presented graphically in Fig.~\ref{partialdensmat} between the excitations.  Note that immediately, one can apply the results from Thm.~\ref{metric_diffE} onto the partial density matrices.  For those states on the left of the orthogonality center, the density matrix is written in the states for the left and right basis functions
\begin{equation}
\rho=\begin{cases}
 U  D^2 U^\dagger&(\mathrm{left})\\
 V  D^2 V^\dagger&(\mathrm{right})
\end{cases}
\end{equation}
and the indices on each tensor are skipped here for the sake of brevity.

In order to use Eq.~\eqref{energydiff} with the density matrix, the lattice sites on the whole system must be partitioned into a left (L) and right (R) group ({\it i.e.}, either $({i\mathcal{A}})_L$ or $({i\mathcal{A}})_R$). The sum can then be separated into each group, and the density matrix must be represented in a basis natural to each group ({\it i.e.} with the basis functions for the left or right as in Def.~\ref{partialdens_def}.
\begin{align}\label{energydiff_partial}
\Delta E^{(m)}_{\alpha\beta;\gamma}&=\sum_{j\in\{({i\mathcal{A}})_L\}}C_j\langle j|U_\gamma\left(D_{\alpha;\gamma}^2- D_{\beta;\gamma}^2 \right) U_\gamma^\dagger|j\rangle\\
&=\sum_{j\in\{({i\mathcal{A}})_R\}}C_j\langle j|V_\gamma\left(D_{\alpha;\gamma}^2 - D_{\beta;\gamma}^2\right)V_\gamma^\dagger|j\rangle\nonumber
\end{align}
where the subscript denotes which excitation the component of the SVD belongs to for some chosen basis $\gamma$ of which $m$ states were kept. Just as in Eq.~\eqref{energydiff}, the Hamiltonian must be local for this expression to be valid. In the tensor network formalism, this corresponds to the MPO containing only local terms. This is because the states $k$ from Eq.~\eqref{energydiff_partial} can also index the basis vectors of the SVD used to compose $U$, $D$, and $V$. This proof applies to any bond in the MPS, so there is no issue of how the MPS is gauged.
\end{proof}

Note that the $D$ matrices for either the $\alpha$ or $\beta$ excitation are guaranteed to be diagonal by construction and are written in the same basis of entangled states. This is exactly the ultra-local limit from Def.~\ref{ultralocalH} that was used in Thm.~\ref{metric_diffE}. In this case, the orbitals $\{\Phi^{(\alpha)}_k\}$ and $\{\Phi_k^{(\beta)}\}$ are the same basis. The degree of freedom that allows one to write a higher excitation is the changing of the magnitude of the occupancies $\varepsilon_k$ for either state. Thus, the general prescription defined for the general case of natural orbitals in the quantum chemistry context is precisely the relevant case for the tensor network.

For smaller energy differences, the most relevant singular vectors will be more common when the states are similar. This means that the primary results from Thm.~\ref{metric_diffE} can be equally applied to the singular vectors here implying a maximum bound on the number of differences in $D$. The singular vectors are also orthogonal in this case, so this is different from Thm.~\ref{metric_diffE} in that the overlap of the singular vectors is not needed.

The reduction here applies equally well for pure or mixed states, although in a given computation the states computed are a mixed state representation of a given eigenvalue. Combined with the previous results, the main statement here applies to all local Hamiltonians.

We are then ready to state a central theorem to the paper.
\begin{theorem}[Bundled tensor networks]
A bundle of states in a tensor network containing an energy interval $[E_\alpha,E_\beta]$ ($\Delta E_{\alpha\beta}\approx0$) and with $\delta^{(\alpha)}\ll1$ and $\delta^{(\beta)}\ll1$ will not require a larger bond dimension $m$ to add a similar state $\delta^{(\gamma)}\ll1$ with $\Delta E_{\alpha\gamma}\approx0$. It will cost at most $m'$ more states (the number required to describe $\gamma$ with $\delta^{(\gamma)}$ of some magnitude) to add a dissimilar state of any energy difference.
\end{theorem}
This immediately follows from the results in Sec.~\ref{partialdensmats} since the energy of the new state on the interval will have a small energy difference with the other states on the compact interval. The summary statement is that the bundled MPS is not necessarily more expensive to solve for if a suitable initial state is solved and then low energy excitations are found.

\subsection{Discussion of states with different symmetries}

How many excitations are similar is highly model dependent. While it is true that the next excitation does not need to be similar, it is true that the set of all states on a compact interval contains similar states in most models. So, there should be an expected reduction in the size of the bond dimension.
 
 Consider two classical MPSs ($m=1$), one with a state with 1 fermion and another with 10 fermions. The overlap between these two states will definitely zero because of the quantum number symmetries. As the bond dimension of the MPS is grown (as in a DMRG computation) and each tensor contains more blocks each with different symmetries, it is not guaranteed that the symmetry sectors will have much overlap upon contraction of the full network. From this argument, it can be reasoned that based on symmetries alone that the singular values will not have much overlap between the two excitations. This is even true for a state with 1 fermion and another with 2 fermions with all fermions on the first site and if the net fluxes of the quantum number symmetries are assigned on the last site.

Conservatively, we can state that the addition of a new state into the bundled MPS that does not have a symmetry represented will require another $m$ states to be added onto the link index. So, the scaling appears as $O(mN_\mathrm{sym.})$; however, this a generous upper bound because of similarities between some symmetry sectors.

Take for example the ensemble containing 1 spin-up fermion and another excitation for 1 fermion with a spin-down fermion for a spin-symmetric Hamiltonian ({\it i.e.}, no magnetic field or other effect applied). These two symmetry sectors in the Hamiltonian are identical blocks by definition. The wavefunctions, however, while not having the same symmetry have many elements in the Hamiltonian in common and have a conjugate set of singular vectors in the other.

Another example would be for a chain of entirely up spins and another state of entirely down spins. The quantum number symmetries on each link index show that the spaces in which the tensors exist is completely orthogonal to each other in general.

\begin{figure*}
\begin{center}
\includegraphics[width=0.66\columnwidth]{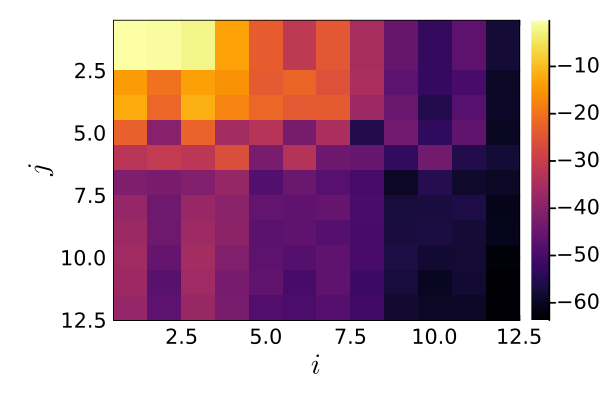}
 \put (-35,25) {\textcolor{white}{a)}}
\includegraphics[width=0.66\columnwidth]{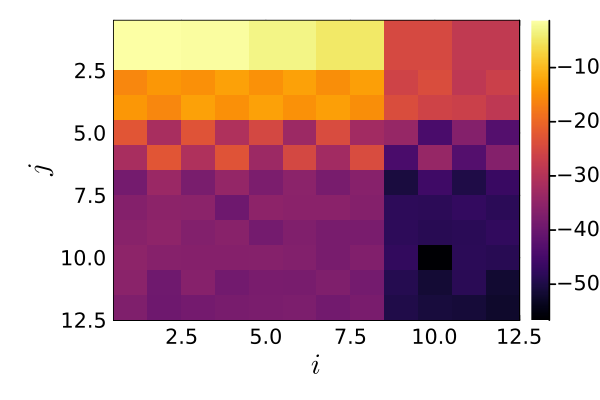}
 \put (-35,25) {\textcolor{white}{b)}}
\includegraphics[width=0.66\columnwidth]{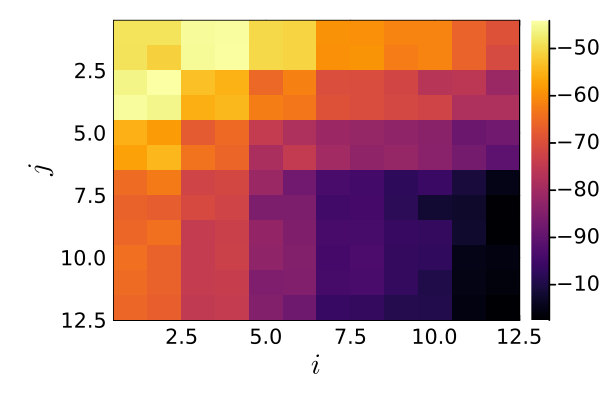}
 \put (-35,25) {\textcolor{white}{c)}}\\
\includegraphics[width=0.66\columnwidth]{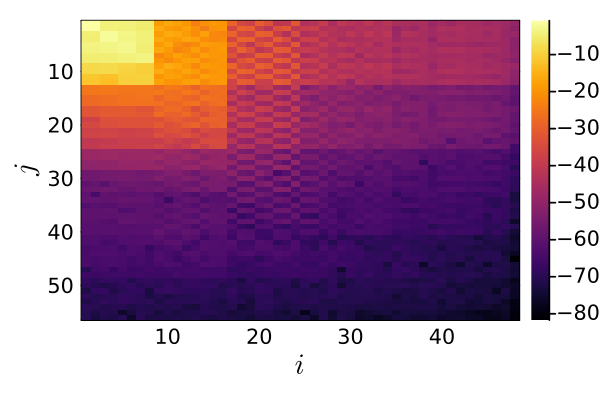}
 \put (-35,25) {\textcolor{white}{d)}}
\includegraphics[width=0.66\columnwidth]{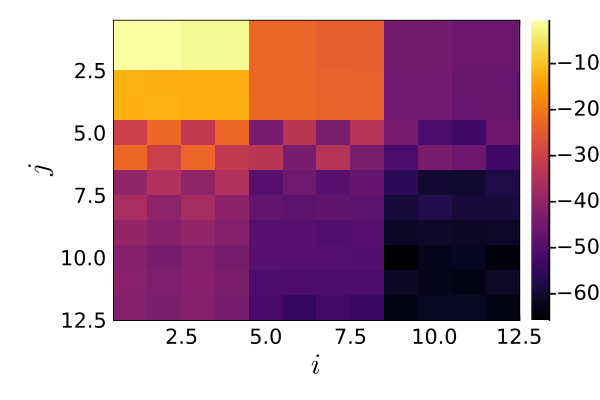}
 \put (-35,25) {\textcolor{white}{e)}}
\end{center}
\caption{\label{naturalorbitalsplot}
Plots of the logarithm of the absolute value of the weights overlap matrix for a 12-site Ising model with $h_x=0.01$ and examining the middle bond of the system. Only non-zero singular values were kept in the SVD but is untruncated otherwise. Brighter colors represent higher values and therefore more significant states for the ground-state. a) A bundle formed of a bundle with states 1 and 2 compared with a bundle of states 1 and 3. Depending on the cutoff tolerated, we can truncate many of the states. b) weighted overlap matrix of a bundle of states 1 and 2 with a bundle of states 29 and 30. More states have higher values because of the increase in norm between the density matrices. Some of the orbitals can still be truncated (darker regions) depending on the allowed cutoff tolerance. c) the overlap matrix of states 1 and 2 with states $M-1$ and $M=2^{12}$. More of the states have a high weight making ti harder to find states to truncate. d) states 28 and 29 bundled together and observed in the overlap matrix with a bundle of 29-30 which have a small energy difference. Many states can again be truncated over many orbitals. e) Two bundles with states $M-1$ and $M$ bundled with states $M-2$ and $M$ show relatively few states of high weight.}
\end{figure*}

\begin{figure*}
\begin{center}
\includegraphics[width=0.66\columnwidth]{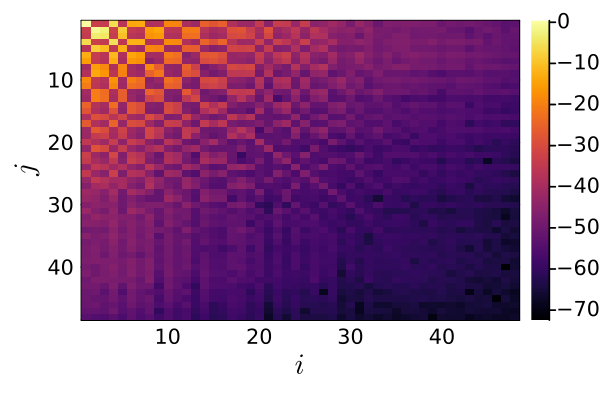}
 \put (-35,25) {\textcolor{white}{a)}}
\includegraphics[width=0.66\columnwidth]{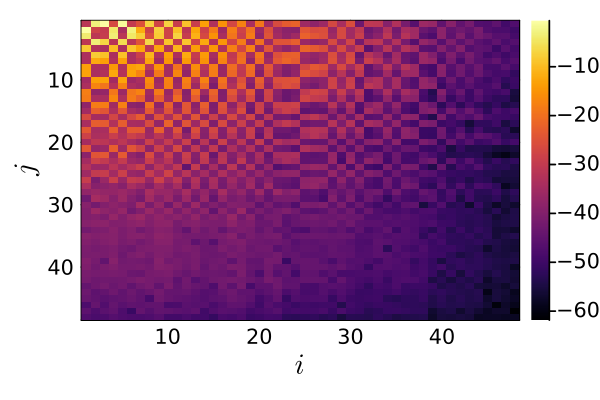}
 \put (-35,25) {\textcolor{white}{b)}}
\includegraphics[width=0.66\columnwidth]{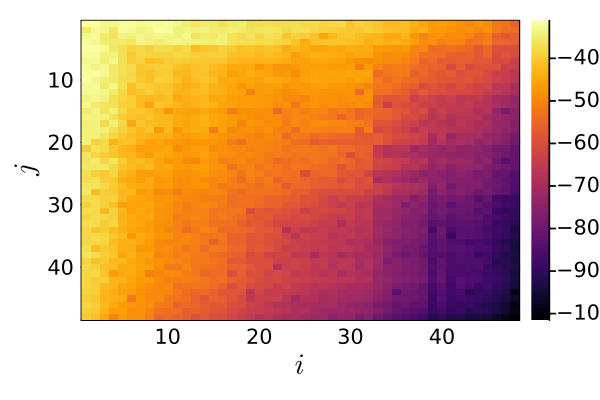}
 \put (-35,25) {\textcolor{white}{c)}}\\
\includegraphics[width=0.66\columnwidth]{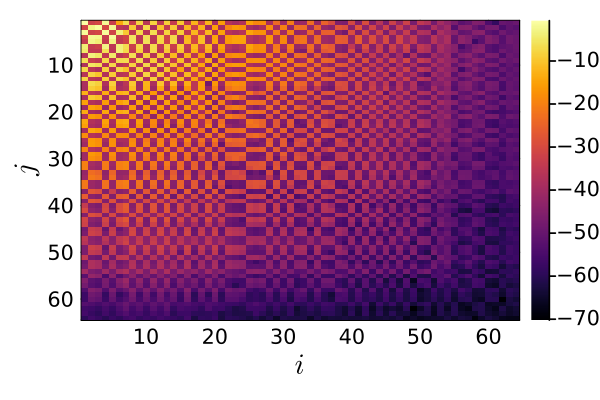}
 \put (-35,25) {\textcolor{white}{d)}}
\includegraphics[width=0.66\columnwidth]{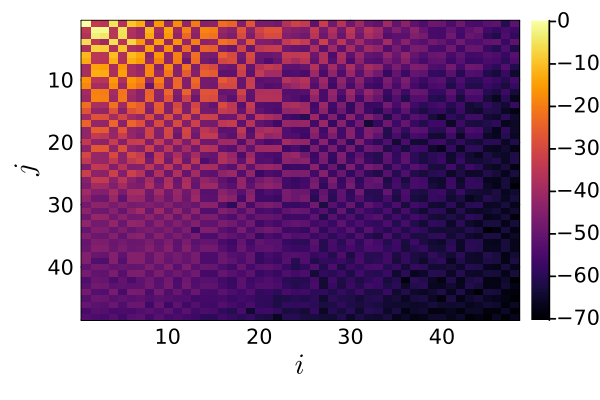}
 \put (-35,25) {\textcolor{white}{e)}}
\end{center}
\caption{\label{naturalorbitalsplot_critical}
The same plots as Fig.~\ref{naturalorbitalsplot} but for an Ising model with $h_x=1$ (critical Ising model). The larger bond dimension is due to the critical nature of the model. More states are necessary because of the long-range correlations in the system, but the qualitative features of the graphs remain the same. Larger energy differences correspond to more states of high weight in this measure and therefore a larger bond dimension is required to keep the truncation error low.}
\end{figure*}

\section{Examples}

To illustrate some of the concepts presented in the previous sections, we focus on three models and determine some relevant quantities. All models are solved with the DMRjulia library.

\subsection{Transverse field Ising model}

We consider the case of the transverse field Ising model, defined as
\begin{equation}
H=\sum_i\sigma_i^z\cdot\sigma_{i+1}^z+h_x\sigma^x_i
\end{equation}
where
\begin{equation}
\sigma^x=\left(\begin{array}{cc}
0 & 1\\
1 & 0
\end{array}\right)\quad\mathrm{and}\quad\sigma^z=\left(\begin{array}{cc}
1 & 0\\
0 & -1
\end{array}\right)
\end{equation}
with subscripts indicating the matrix belongs to a Pauli string of Kronecker products $O_i=I\otimes I\otimes\ldots\otimes O\otimes\ldots\otimes I$ with $O$ in the $i$th position. The model exhibits a phase transitions when $h_x=1$, which creates a gapless spectrum \cite{sachdev2007quantum}. We note that we focus on the properties of the bundled MPS in this paper and not the algorithm used to solve it, which in the most general case should be considered to be undecideable--in the sense of the halting problem--to find the spectral gap of the Hamiltonian \cite{cubitt2015undecidability}, meaning that finding the right method to solve the bundled MPS is an open challenge for arbitrary systems (including three-dimensional models).

One way to visualize the relevant states in the system for different energies is to form the overlap matrix, $\Gamma$ from Fig.~\ref{partialdensmat} and multiply the $D$ matrices contracted onto the link index, giving
\begin{equation}
\Gamma_{ij} = \Gamma_{a_i'a_i}=\sum_{a_{i-1},a_i,a'_{i-1},a'_i}\rho^{(L)}_{a_{i-1}a_{i-1}'}D_{a_{i-1}a_i}\tilde D_{a_{i-1}'a_i'}
\end{equation}
This ensures that no only will we be able to see the overlap between two of the basis states used in two different bundles, but we will also see how relevant each are for the overall answer. All plots are shown with the logarithm of the absolute value of each element of the overlap matrix $\Gamma$. Matrices are truncated so that no singular values of zero are represented. 

Brighter colours indicate that the matrix element is important and should not be truncated. Darker colours indicate that the matrix elements is small. Truncating an entire row or column is possible in the normal evaluation of the reduced bond dimension of a system. Thus, we look in the figures for a row or column where (nearly) all of the elements are a darker colour. This means that the row or column can be removed without losing precision in the bundled MPS.

The results in Fig.~\ref{naturalorbitalsplot} tell the same story as the theorems presented previously. The larger the energy difference between states, the more states in this measure have a higher weight (lighter color). 

The general trend is that the larger the energy difference, the fewer rows and columns that can be truncated.

Results for the critical model in Fig.~\ref{naturalorbitalsplot_critical} are largely the same but the bond dimension is larger because of the gapless eigenspectrum of the model.

\subsection{Heisenberg model}

We can perform the same analysis with an XXZ model ($\Delta=1$) of the form
\begin{equation}
H=\sum_i\mathbf{S}_i\cdot\mathbf{S}_{i+1}
\end{equation}
where $\mathbf{S}=\langle S^x,S^y,S^z\rangle=\frac12\langle \sigma^x,\sigma^y,\sigma^z\rangle$ and the commutator $[\sigma^x,\sigma^z]=-2i\sigma^y$ \cite{townsend2000modern}.

\begin{figure*}
\begin{center}
\includegraphics[width=\columnwidth]{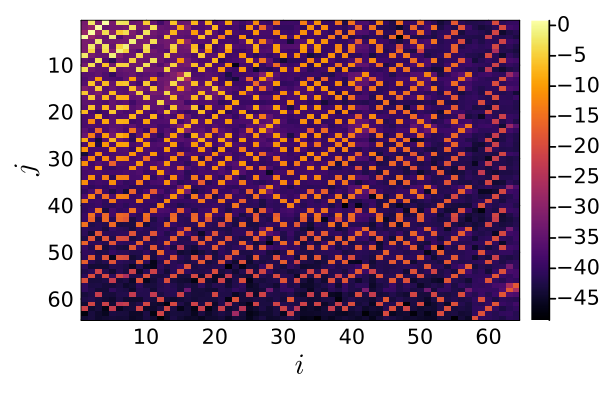}
 \put (-50,40) {\textcolor{white}{a)}}
\includegraphics[width=\columnwidth]{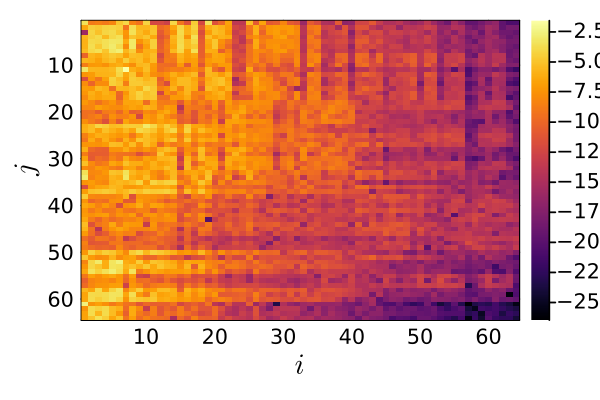}
 \put (-50,40) {\textcolor{white}{b)}}\\
\includegraphics[width=\columnwidth]{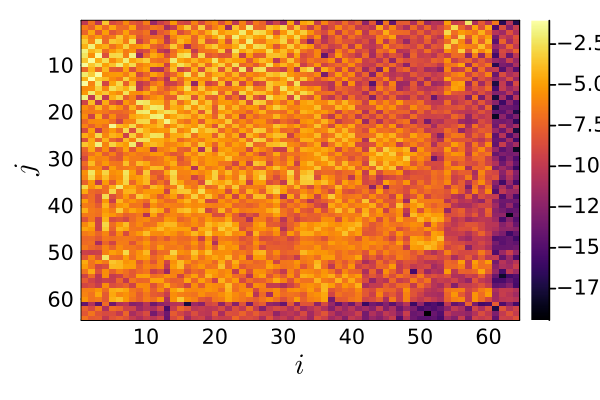}
 \put (-50,40) {\textcolor{white}{c)}}
\includegraphics[width=\columnwidth]{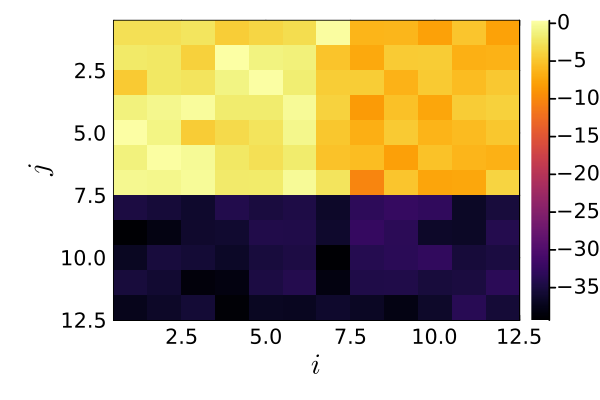}
 \put (-50,40) {\textcolor{white}{d)}}
\end{center}
\caption{\label{naturalorbitalsplot_XXZ}
A 12-site Heisenberg model ($M=2^{12}$). This time, we bundle together groups of 10 states into a bundle. a) The eigenstates 1-10 and 11-20 with the computed weighted overlap matrix. Darker rows and columns indicate that the plot can be truncated up to some defined tolerance. b) The first ten (1-10) excitations computed in the weighted overlap matrix with last $M-10$ to $M$ excitations. There are far fewer states which can be truncated in this case. Note that the left side of the plot is all high values. c) states $M/2-(M/2+10)$ and $(M/2+11)-(M/2+20)$. More states can be trucnated than for the largest energy difference shown (b) but more high-weight values appear because it is more costly to solve for the lowest energy in this case. d) For the states $(M-10)-M$ and $(M-20)-(M-11)$ show that the number of states is small again. Note that many states were truncated in comparison to the other three figures because there were many weight zero singular values that were truncated.}
\end{figure*}

Figure~\ref{naturalorbitalsplot_XXZ} shows the weighted overlap matrix. This time, groups of 10 eigenstates are shown. The Heisenberg model has a well-known symmetry structure that creates a blockier appearance for the graph. However, the main message of the analysis remains the same: larger energy differences have more high-weight states. Bundling together energy states that are close together have many states in common and thus the bond dimension can be truncated.

\section{Conclusion}

We introduced the bundled density matrix, a set of density matrices that are independent but written in a common basis. We showed that the difference in the truncation errors constitute a practical metric to determine a notion of distance between the density matrices. It was shown that for local system, the energy difference in that basis was also a practical metric to describe a notion of distance between density matrices in the bundle.

The larger the energy difference between the density matrices, the larger the difference in the density matrices are. One might expect that as one bundles excitations that are further into the bulk that the volume law entanglement would dominate and therefore drive the bond dimension higher. What these results suggest, in effect, is that this is only guaranteed for large energy differences. For small energy differences, the density matrices can either be similar, where the $m$ most relevant states constitute a good basis for low-energy excitations. Or, the states can be dissimilar, such as in the case of a different symmetry sector, where the truncation error is large.

We extended these ideas to the bundled matrix product state where the results demonstrate that the bond dimension is lower for similar states. Bundled matrix product states therefore may support low energy excitations in the same symmetry sector without a costly increase in the bond dimension. To add a dissimilar state, one must pay the cost of the matrix product state representation of that state. Large energy differences in bundled matrix product states have effectively no degrees of freedom overlapping and are not very different from two separate MPSs.

Definitions for similar states and the exact applicability of the local limit used in this paper can certainly be expanded and modified, and we encourage the community to work with the concept of the bundled density matrix in other contexts including quantum information.

\section{Acknowledgements}

This research was undertaken, in part, thanks to funding from the Canada Research Chairs Program.

The Chair position in the area of Quantum Computing for Modelling of Molecules and Materials is hosted by the Departments of Physics \& Astronomy and of Chemistry at the University of Victoria.

N.S.~acknowledges the NSERC CREATE in Quantum Computing Program (Grant Number 543245).

This work is supported by a start-up grant from the Faculty of Science at the University of Victoria. 

This work has been supported in part by the Natural Sciences and Engineering Research Council of Canada (NSERC) under grants RGPIN-2023-05510 and DGECR-2023-00026.

\begin{appendix}

\section{Two equivalent forms to determine the density matrix}\label{densitymatrixequivalence}

The definition of the one-body reduced density matrix is 
\begin{align}
\rho(\mathbf{x},\mathbf{x}')=&\int\ldots\int\psi^*(\mathbf{x},\mathbf{r}_2,\ldots,\mathbf{r}_{N_e})\\
&\quad\times\psi(\mathbf{x}',\mathbf{r}_2,\ldots,\mathbf{r}_{N_e})d\mathbf{r}_2\ldots d\mathbf{r}_{N_e}\nonumber
\end{align}
which works well for applications in quantum chemistry. In the graphical notation used in the main section of this paper, the vertical physical indices would have to be kept track of as separate but all other physical indices would be contracted. 

The case of the tensor network has an immediately large computational cost with using this form. The tensor network form would require that we isolate physical indices and perform a contraction that winds up being exponential for large $|i-j|$. Neither of these is desirable for the tensor network because it was specifically formulated not incur the exponential cost of the full quantum problem.

Another definition of the density matrix with elements
\begin{equation}
\rho_{ij}=\langle\psi| c^\dagger_{i\sigma} c_{j\sigma}|\psi\rangle
\end{equation}
for fermions \cite{dmrjulia1} is often cited. This form is equivalent to the above form as we now demonstrate.

Consider the definition of the wavefunction as a superposition in some position. We denote the local degrees of freedom on each cite $i$ by $\sigma_i$. The general decomposition of the wavefunction becomes
\begin{equation}\label{psiform}
|\psi\rangle=\sum_{\sigma_1\dots\sigma_N}w_{\sigma_1\dots\sigma_N}|{\sigma_1\dots\sigma_N}\rangle
\end{equation}
for $N$ sites. The application of the operator $ c_j$ gives
\begin{equation}\label{ketform}
 c_j|\psi\rangle=\sum_{\sigma_1\dots\sigma_N}w_{\sigma_1\dots\sigma_N}|{\sigma_1\dots\bar\sigma_j\ldots\sigma_N}\rangle
\end{equation}
where $\bar\sigma_j$ is the modified value of $\sigma_j$ upon evaluation of the operator. In the case of a spin-half system, a value of $\uparrow$ would go to $\downarrow$.

A similar form can be derived for $ c^\dagger_i$ on the dual vectors
\begin{equation}
\langle\psi| c^\dagger_i=\sum_{\sigma_1\dots\sigma_N}\langle{\sigma_1\dots\bar\sigma_i\ldots\sigma_N}|w_{\sigma_1\dots\sigma_N}^*
\end{equation}
Neither this sum nor the dual expression in Eq.~\eqref{ketform} have the same number of terms as in Eq.~\eqref{psiform} because the annihilation operator creates some zero terms from the original expression. The final expectation value becomes
\begin{equation}\label{latticeform}
\langle\psi| c^\dagger_i c_j|\psi\rangle=\sum_{\sigma_i\sigma_j}\langle\bar\sigma_i\sigma_j|W_{\sigma_i\sigma_j}|\sigma_i\bar\sigma_j\rangle.
\end{equation}
where
\begin{equation}
W_{\sigma_i\sigma_j}=\sum_{\{\sigma_1\dots\sigma_N\}\backslash\{\sigma_i\sigma_j\}}w^*_{\sigma_1\ldots\sigma_N}w_{\sigma_1\ldots\sigma_N}
\end{equation}
where the summation is over all variables but not $\sigma_i$ or $\sigma_j$.
Returning now to Eq.~\eqref{psiform}, it can be seen that the replacement of the $\mathbf{x}$ and $\mathbf{x}'$ terms would be represented with the ket form as
\begin{align}\label{truedensmat}
\rho(\mathbf{x},\mathbf{x}')\Rightarrow&\rho(\sigma_\mathbf{x},\sigma_{\mathbf{x}'})\overset{\mathrm{discretize}}=\rho(\sigma_i,\sigma_j)\\
=&\sum_{\{\sigma_1\dots\sigma_N\}\backslash\{\sigma_i\sigma_j\}}\langle\sigma_1\dots\sigma_N|W_{\sigma_i\sigma_j}|{\sigma_1\dots\sigma_N}\rangle\nonumber
\end{align}
and the primed index on $i$ or $j$ is a reinterpretation of the continuous, real-space position variable $\mathbf{x}$ into the lattice variable for some site on the discrete lattice space. Eq.~\eqref{truedensmat} then reduces to Eq.~\eqref{latticeform}. Thus, the two forms of the density matrix are the same.

\end{appendix}

\bibliography{TEB_papers,TEB_books,refs}

\end{document}